\documentclass[letterpaper, 10 pt, conference]{ieeeconf}  

\IEEEoverridecommandlockouts                              

  \let\labelindent\relax
\usepackage{amsmath} 
\usepackage{amssymb}  
\usepackage{cite}

\usepackage{tabulary}
\usepackage[english]{babel}
\usepackage{blindtext}
\usepackage[makeroom]{cancel}
\usepackage{amsfonts}
\usepackage{amsthm}  
\usepackage{amssymb}
\usepackage[pdftex]{graphicx}
\graphicspath{{../pdf/}{../jpeg/}}
\DeclareGraphicsExtensions{.pdf,.jpeg,.png}
\usepackage{pgf,tikz}
\usepackage{tabularx}
\usepackage{float}

\usepackage{enumitem} 
\usepackage{bbm} 

\usepackage[strict]{changepage}

\newcommand{\vect}[1]{\boldsymbol{#1}}

\newcommand{\wt}[1]{\widetilde{#1}}

\newtheorem{remark}{Remark}
\newtheorem{thm}{Theorem}
\newtheorem{lemma}{Lemma}

\title{\LARGE \bf
Nonlinear Mapping Convergence and Application to Social Networks
}

\author{  Brian D.O. Anderson$^{1,2}$, Mengbin Ye$^{1}$,
\thanks{$^1$ B.D.O. Anderson and M. Ye are with the Research School of Engineering, Australian National University
\texttt{\{Mengbin.Ye, Brian.Anderson \}@anu.edu.au}. $^2$B.D.O. Anderson is also with Hangzhou Dianzi University, Hangzhou, China 
and with Data61-CSIRO (formerly NICTA Ltd.) in Canberra, A.C.T., Australia. 
}}

\begin{document}

\maketitle
\thispagestyle{empty}
\pagestyle{empty}

\begin{abstract}
This paper discusses nonlinear discrete-time maps of the form
$x(k+1)=F(x(k))$, focussing on equilibrium points of such maps. Under some circumstances, Lefschetz fixed-point theory can be used to establish the existence of a single locally attractive equilibrium (which is sometimes globally attractive) when a general property of local attractivity is known for any equilibrium. Problems in social networks often involve such discrete-time systems, and we make an application to one such problem.  
\end{abstract}

\section{Introduction}\label{sec:intro}
%
%
%
%

Recursive equations of the form \begin{equation}\label{eq:nonlinearupdate}
x(k+1)=F(x(k)),
\end{equation}
are fundamental to control and signal processing. Very often $F$ is linear or affine, but in this paper, $F$ is not so restricted, though we do require it to be suitably smooth. Usually also, $x(k)$ resides in a Euclidean space of known dimension, though this is not always the case, and indeed will not always be the case in this paper. 

In many situations, it is possible to examine local behavior of the nonlinear map \eqref{eq:nonlinearupdate} around an equilibrium point, through a linearization process. If $\bar x$ is an equilibrium point, i.e. a fixed point of the mapping $F$ satisfying $\bar x = F(\bar x)$, then the Jacobian $J(\bar x)=\frac{\partial F}{\partial x}|_{\bar x}$ provides guidance as to behavior in the vicinity of $\bar x$. If $||x(k)-\bar x||$ is small, then approximately
\begin{equation}\label{eq:linearizationupdate}
x(k+1)-\bar x=J(\bar x)[x(k)-\bar x]
\end{equation}
If the eigenvalues of $J(\bar x)$ do not lie on the unit circle, then the asymptotic stability or instability of the linear equation \eqref{eq:linearizationupdate} implies the same property for the nonlinear equation \eqref{eq:nonlinearupdate}, albeit locally. 

Recent work in the area of social networks \cite{jia2015opinion_SIAM} introduced what amounts to a particular version of \eqref{eq:nonlinearupdate}, and established by a rather specialized calculation, tailored to the specific algebraic form of $F$, that under normal circumstances, the equation possessed a single globally attractive equilibrium. For completeness, we note that in the application, the entries of $x(k)$ were restricted to lie in $[0,1]$ and to satisfy $\sum_jx_j(k)=1,\;\forall k$. 

It is natural to speculate whether the conclusion that there is a single attractive equilibrium is indeed intrinsic to the algebraic form of $F$, or whether rather, it is a consequence of some more general property, and consequently also one that will follow for a whole class of $F$ of which that in \cite{jia2015opinion_SIAM} is just a special case. 


Indeed we show that the conclusion is a general one, and we use Lefschetz fixed point theory to show this. The method advanced here may have application to other situations than those considered in the paper \cite{jia2015opinion_SIAM}, since they are more general in character, i.e. will apply in a much wider variety of situations than those contemplated in the paper \cite{jia2015opinion_SIAM}.

By way of brief background, Lefschetz fixed point theory (of which more details are summarized subsequently) is a tool for relating the local behavior of maps to some global properties, taking into account the underlying topological space in which the maps act. The local properties are associated with the linearized equations \eqref{eq:linearizationupdate}, potentially studied at multiple equilibrium points (and with in general a different $J(\bar x)$ associated with each equilibrium point). Such local properties were flagged in \cite{ye2017DF_journal_arxiv} as of central concern in a time-varying version of the problem studied in \cite {jia2015opinion_SIAM}.



To sum up the contribution of this paper, we provide a new result demonstrating local exponential convergence to a unique fixed point, for nonlinear maps known to have local convergence properties around its equilibria, and we indicate the applicability of the result to a problem in social network analysis where global convergence occurs.  

The paper is structured as follows. In the next section, we present background results on Lefschetz fixed point theory. While these may be standard to those familiar with differential topology, they are not so standard for control engineers. Section \ref{section:mainresult}  presents the main result of the paper, indicating circumstances under which there is a unique equilibrium point for \eqref{eq:nonlinearupdate} which is in fact locally exponentially stable. The following section illustrates the application to the update equation arising in social networks \cite{jia2015opinion_SIAM}, and section \ref{section:concludingremarks} contains concluding remarks.

\section{Background on Lefschetz Fixed Point Theory}

Lefschetz fixed-point theory applies to smooth maps $F:X\rightarrow X$ where $X$ is a compact oriented manifold \cite{guillemin2010differential}, \cite{hirsch2012differential} or a compact triangulable space \cite{armstrong2013basic}.{\color{blue}}{\footnote{The notion of orientation of a manifold is described in the references; roughly, a manifold is oriented if one can attach an infinitesimal set of coordinate axes to an arbitrary point on the manifold, and then move the point with the axes attached knowing that one can never move to reverse the orientation. A M\"obius strip is \textit{not} an oriented manifold.}} Thus $X=\mathbb R^n$ is excluded, but if $X$ is a compact subset of $\mathbb R^n$ such as a simplex, then it is allowed. This also means that if a map $F:\mathbb R^n\rightarrow\mathbb R^n$ is known to have no fixed points for large values of its argument, the theory can be applied by considering the restriction of $F$ to a compact subset of $\mathbb R^n$ such as a ball of large enough radius. 

Lefschetz fixed-point theory involves derivatives. Any smooth map has the property that at any point $x\in X$, there is a linear derivative mapping, call it $dF_x$, and if $X$ looks locally like $\mathbb R^m$,  then the derivative map can be represented by the $m\times m$ Jacobian matrix in the local coordinate basis. 

Interest is centered for our purposes on those maps which have a finite number of fixed points (including possibly zero) in $X$, though of course, maps with an infinite set of fixed points exist, for example $F(x)=x$, the identity map. A fixed point $x$ is called a \textit{Lefschetz fixed point of F} if the eigenvalues of $dF_x$ are unequal to 1. A fixed point being a Lefschetz fixed point is sufficient but not necessary to ensure that $x$ is an isolated fixed point of $F$, i.e. there is an open neighborhood around $x$ in which no other fixed point occurs. Because $X$ is compact, and if it is known that all fixed points of $F$ are isolated, say because they are all Lefschetz fixed points, it easily follows that  the number of fixed points is necessarily finite. For completeness, we record an argument by contradiction, which seems standard. If there were an infinite number of fixed points, $x_i, i=1,2,\dots$, compactness of $X$ implies there is a convergent subsequence $x_{i1}, x_{i2},\dots$, with limit point $\bar x$, and again by compactness $\bar x\in X$. Now $F$ is continuous so $F(x_{ij})\rightarrow F(\bar x)$ since $x_{ij}\rightarrow \bar x$ as $j\rightarrow\infty$. Then $x_{ij}-F(x_{ij})
\rightarrow \bar x-F(\bar x)$ as $j\rightarrow\infty$. Since $x_{ij}-F(x_{ij})=0\;\forall j$, it is evident that $\bar x$ is a fixed point of $F$. However, being a limit point it is not isolated, hence the contradiction.

The Lefschetz property holding at a particular fixed point $x$ also implies that at the point $x$, the (linear) mapping $I-dF_x$ is an isomorphism of the tangent space $T_x(X)$ at $x$. If it preserves orientation, then its determinant is positive, while if it reverses orientation, its determinant is negative. The \textit{local Lefschetz number} of $F$ at a fixed point $x$, written $L_x(F)$, is defined as +1 or -1 according as the determinant of $I-dF_x$ is positive or negative.\footnote {Reference \cite{guillemin2010differential} uses $dF_x-I$ rather than $I-dF_x$, which is used by \cite{hirsch2012differential}. We require the latter form.}

The map $F$ is termed a \textit{Lefschetz map} if and only if all its fixed points are Lefschetz fixed points (and there are then, as noted above, a finite number of fixed points). The \textit{Lefschetz number of $F$}, written $L(F)$,  is defined as

\begin{equation}\label{eq:Lefschetzsum}
L(F)=\sum_{F(x)=x}L_x(F)
\end{equation}

There is an alternative definition of the Lefschetz number not provided here which can be shown to be equivalent to that appearing here, based on topological considerations, and provided in \cite{hirsch2012differential, guillemin2010differential}. It is \textit{not} restricted to maps with a finite number of fixed points. Moreover, using this alternative definition, one sees that $L(F)$ is a \textit{homotopy invariant}, \footnote{Smooth maps $F:X\rightarrow X$ and $G: X\rightarrow X$ are said to be homotopic if there exists a smooth map $H:X\times I\rightarrow X\times I$ with $H(x,0)=F(x), H(x,1)=G(x)$. Saying $L(F)$ is a homotopy invariant means $L(F)=L(G)$ for any $G$ which is homotopic to $F$.} and this particular property does not require limitation to those maps with a finite number of fixed points. Further, the alternative approach yields a connection between the Lefschetz number of the identity map (which has an infinite number of fixed points) and another topological invariant, of the underlying space $X$, viz the Euler characteristic{\footnote{The Euler characteristic is an integer number associated with a topological space, including a space that in some sense is a limit of a sequence of multidimensional polyhedra, e.g. a sphere, and a key property is that distortion or bending of the space leaves the number invariant. Euler characteristics are known for a great many topological spaces.}}, \cite{hirsch2012differential, guillemin2010differential,matsumoto2002introduction}.

The key result (see e.g. \cite{hirsch2012differential} for the case of a compact oriented manifold and \cite{armstrong2013basic} for the case of a compact triangulable space) is as follows:
\begin{thm}\label{thm:Lefschetz}
The Lefschetz number of the identity map $\mathcal I_d: X\rightarrow X$ where $X$ is a compact oriented manifold or a compact triangulable space is $\chi(X)$, the Euler characteristic of $X$. 
\end{thm} 
A key consequence of this theorem is that if a map $F$ is homotopically equivalent to $\mathcal I_d$, i.e. if there exists a smooth map $H:X\times I\rightarrow X$ such that $\hat F(x,0)=F(x)$ and $H(x,1)=\mathcal I_d(x)=x$  then 
\begin{equation}
L(F)=\chi(X)
\end{equation}
Hence we have the following theorem:
\begin{thm}\label{thm:Lefschetz2}
Let $X$ be a compact oriented manifold or a compact triangulable space, and suppose $F:X\rightarrow X$ is a Lefschetz map, i.e. there is a finite number of fixed points at each of which $I-dF_x$ is an isomorphism, and is homotopically equivalent to the identity map. Then there holds
\begin{equation}
L(F)=\sum_{F(x)=x}L_x(F)=\chi(X)
\end{equation}
where $L_x(F)$ is $+1$ or $-1$ according as $det(I-dF_x)$ has positive or negative sign, and $\chi(X)$ is the Euler characteristic of $X$.
\end{thm}

\section{Main Result}\label{section:mainresult}

In this section, we establish that certain properties of the mapping $F$ and the associated space $X$ guarantee that $F$ has a unique fixed point. The main result, proved using the Lefschetz theory, is as follows.

\begin{thm}\label{thm:main_lefschetz_result}
Consider a smooth map $F:X\rightarrow X$ where $X$ is a compact, oriented and convex manifold or a convex triangulable space of arbitrary dimension. Suppose that the eigenvalues of $dF_x$ have magnitude less than 1 for all fixed points of $F$. Then $F$ has a unique fixed point, and in a local neighborhood about the fixed point, \eqref{eq:nonlinearupdate} converges to the fixed point exponentially fast.
\end{thm}
\begin{proof}
Observe first that the compactness and convexity  properties of $X$ guarantee it is homotopy equivalent to the unit $m$-dimensional disk $D^m$ and accordingly then homotopy equivalent to a single point. This means that $\chi(X)=1$, see e.g.\cite{matsumoto2002introduction}, see p. 140. 

Next, observe that, because $X$ is convex, $H =t\mathcal I_d+(1-t)F$ which maps $x$ to $tx+(1-t)F(x)$, is a mapping from $X$ to $X$ for every $t\in[0,1]$ and the smoothness properties of $H$ (which come from the smoothness of $F$ and the specific dependence on $t$) then guarantee that $F$ and $\mathcal I_d$ are homotopically equivalent. By Theorem \ref{thm:Lefschetz2}, there holds
\begin{equation}\label{eq:LefschetzEuler}
L(F)=1
\end{equation}
Now for any real matrix $A$ for which the eigenvalues are less than one in magnitude, it is easily seen that the matrix $I-A$ has eigenvalues all with positive real part, from which it follows that the determinant of $I-A$ is positive, since the determinant is equal to the product of the eigenvalues. Hence for any fixed point $x$ of $F$,  we see by identifying $A$ with $dF_x$ that there necessarily holds $L_x(F)=1$, 
By \eqref{eq:Lefschetzsum} and \eqref{eq:LefschetzEuler}, it follows that 
$$1=\sum_{F(x)=x}1$$
or that there is precisely one fixed point. 

Convergence of \eqref{eq:nonlinearupdate} to the unique fixed point from any initial value in its region of attraction is necessarily exponentially fast. In a neighborhood $\mathcal{D}$ around the unique fixed point, the eigenvalue property of $dF_x$ guarantees exponential convergence. The region of attraction for the fixed point is in most instances larger than $\mathcal{D}$, and we denote as $\mathcal{U} \subset X$ an arbitrary \emph{compact} space within the region of attraction and containing $\mathcal{D}$. For any initial $x\in \mathcal{U}$, the sequence $x, F(x), F(F(x)),\dots$  converges to the neighborhood $\mathcal{D}$ in a finite number of steps, and because the set $U$ is compact, there is a number of steps, $\bar N < \infty$ say, such that from all initial conditions in $\mathcal{U}$, the neighborhood is reached in no more than $\bar N$ steps. The finiteness of $\bar N$ then implies that exponentially fast convergence occurs for all initial conditions $x\in \mathcal{U}$. 
\end{proof}

\begin{remark}
We stress that a key feature of our result is that we need only evaluate the Jacobian $dF_x$ at the fixed points of $F$. In contrast,  recall that a standard method to prove that $F : X \mapsto X$ has a unique fixed point $\bar x$ and that \eqref{eq:nonlinearupdate} converges exponentially fast to $\bar x$ is via Banach's Fixed Point Theorem (assuming $X$ compact). Specifically, one sufficient condition for $F$ to be a contractive map would be to prove that $|| dF_x || < \alpha,\;\forall\, x \in X$, where $\alpha < 1$, with a further assumption that $X$ be convex \cite{khamsi2011metric_space_book}. Thus global properties rather than local (at fixed point) properties are required to generate the conclusion. The difficulty is acute for us. 
A nonlinear $F$ results in $dF_x$ being state-dependent. Consider two consecutive points of the trajectory of \eqref{eq:nonlinearupdate} that are not a fixed point, which we denote $x_1 = x(k)$ and $x_2 = x(k+1)$, and suppose that $dF_x|_{x_1}$ and $dF_x|_{x_2}$ both have eigenvalues with magnitude less than 1, i.e. assume that the eigenvalue restriction applies other than just at the fixed points. Then according to \cite[Lemma 5.6.10]{horn2012matrixbook}, there exists norms $\Vert \cdot \Vert^\prime$ and $\Vert \cdots \Vert^{\prime\prime}$ such that $\Vert dF_x|_{x_1} \Vert^\prime <1$ and $\Vert dF_x|_{x_2} \Vert^{\prime\prime} <1$. However, it cannot be guaranteed that there exists a single norm $\Vert \cdot \Vert^{\prime \prime \prime}$ such that $\Vert dF_x|_{x_1} \Vert^{\prime \prime \prime}, \Vert dF_x|_{x_2} \Vert^{\prime \prime \prime} <1$. In this paper, we need not consider norms; we need not consider eigenvalue properties at all points; we only need to consider the eigenvalues of $dF_x$ at fixed points $\bar x = F(\bar x)$ to simultaneously obtain a unique fixed point conclusion and local exponential convergence.
\end{remark}


\begin{remark} The proof of the theorem using Lefschetz ideas will clearly generalize in the following way. Suppose that $F$ is homotopic to the identity and $X$ is not homotopic to the unit ball, while all fixed points are Lefschetz with the property that $I-dF_x$ has positive determinant. Then the number of fixed points will be $\chi(X)$. If for example $F$ mapped $S^2$ to $S^2$ and never mapped a point to its antipodal point, i.e. there was no $x$ for which $F(x)=-x$, it will be homotopic to the identity map and then there will be two fixed points, since $\chi(S^2)=2$. To construct the homotopy, observe that, because of the exclusion that $F$ can map any point to an antipodal point, there is a well-defined homotopy provided by $$H(x,t)=\frac{(1-t)x+tF(x)}{||(1-t)x+tF(x)||}$$
\end{remark}

\section{Application to a Social Network Problem}\label{sec:social_network}
We now apply the above results to a recent problem in social networks, which studied the evolution of individual self-confidence, $x_i(k)$, as a social network of $n$ individuals discusses a sequence of issues, $k = 0, 1, 2, \hdots$. For simplicity, we consider $n\geq 3$ individuals. We provide a brief introduction to the problem here, including the techniques used to study stability, and refer the reader to \cite{jia2015opinion_SIAM,ye2017DF_journal_arxiv} for details. The map $F$ in question is given as
\begin{align}\label{eq:map_F_DF}
F(  x(k) ) =   \begin{cases} 
   e_i & \hspace*{-6pt} \text{if }  x(k) = e_i \; \text{for any } i \\ \\
   \alpha (x(k)) \begin{bmatrix} \frac{\gamma_1}{1-x_1(k)} \\ \vdots \\ \frac{\gamma_n}{1-x_n(k)} \end{bmatrix}       & \text{otherwise }
  \end{cases}
\end{align}
with $\alpha(x(k)) = 1/\sum_{i=1}^n \frac{\gamma_i}{1- x_i(k)}$ where the vector $\gamma = [\gamma_1, \gamma_2, \hdots, \gamma_n]^\top$ is constant, has strictly positive entries $\gamma_i$ and satisfies $\gamma^\top 1_n = 1$. It can be verified that $F : \Delta_n \mapsto \Delta_n$ where $\Delta_n = \{x_i : \sum_i^n x_i = 1, x_i \geq 0\}$ is the $n$-dimensional unit simplex. Thus, $\Delta_n$ satisfies all the requirements on compactness, orientability, and convexity. Moreover, $F$ is smooth everywhere on $\Delta_n$, including at the corners $x_i=1$, even given  the $1/(1-x_i)$ term in the $i^{th}$ entry of $F$. In \cite{jia2015opinion_SIAM} it is proved that $F$ is continuous using a complex calculation to obtain the Lipschitz constant at the corners of the simplex, but  smoothness is not shown.  As later shown, it follows easily however that $F$ is in fact of class $\mathcal{C}^\infty$ in $\Delta_n$. 

 Let us also make the important point that the above definition \eqref{eq:map_F_DF} of $F$ can be regarded as defining a map $\mathbb R^n\rightarrow\mathbb R^n$, \textit{or} as defining a map on an $(n-1)$-dimensional triangulable space $\Delta_n\rightarrow\Delta_n$, with the $n$-dimensional vector $x=[x_1,x_2,\dots,x_n]^{\top}$ providing a convenient parametrization of the space given imposition of the constraints $\sum_{i=1}^n x_i=1, x_i\geq 0$. 

\begin{remark}It was proved in \cite{jia2015opinion_SIAM} that, in the context of the social network problem, $\gamma_i \leq 1/2$. Since $\gamma_i > 0$ and $n \geq 3$, if $\exists i: \gamma_i = 1/2$ then $\gamma_j < 1/2$ for all $j\neq i$. It was also proved that $\gamma_i = 1/2$ if and only if the graph $\mathcal{G}$, describing the relative interpersonal relationships between the individuals, is a strongly connected ``star graph'' with center node $v_1$. In this paper, we will not consider the special case of the strongly connected star graph.
\end{remark}

\subsection{Existing Results}
In the paper \cite{jia2015opinion_SIAM}, which first proposed the dynamical system \eqref{eq:nonlinearupdate} with map $F$ given in \eqref{eq:map_F_DF}, the following analysis was provided. Firstly, because $F$ is continuous and $\Delta_n$ is convex and compact, Brouwer's Fixed Point Theorem is used to conclude there exists at least one interior fixed point. Next, the authors used a series of inequality calculations, exploiting the algebraic form of $F$, to show that the fixed point  $\bar x$ is unique, and importantly, that $\bar x$ is in the interior of $\Delta_n$. Following this, the authors showed that the trajectories of $x(k)$ had specific properties, again by exploiting the algebraic form of $F$. Lastly, a Lyapunov function is proposed and the properties of the trajectories of $x(k)$ are used to show the Lyapunov function is nonincreasing; LaSalle's Invariance principle is used to conclude asymptotic convergence to the unique interior fixed point $\bar x$ for all initial conditions $x(0)$ that are not a corner of the simplex $\Delta_n$.

The paper \cite{ye2017DF_journal_arxiv} takes a different approach, and looks at the Jacobian of $F$ both as a map $\mathbb R^n\rightarrow \mathbb R^n$ and its restriction (after choice of an appropriate coordinate basis for $\Delta_n$) as a map $\Delta_n\rightarrow\Delta_n$. (Note that in any fixed coordinate basis, the second Jacobian is of dimension $(n-1)\times(n-1)$, with the two Jacobians necessarily related, as described further below. It is this second Jacobian which represents the mapping $dF_x$ defined in earlier sections.)   However, rather than using the results in this paper, \cite{ye2017DF_journal_arxiv} uses nonlinear contraction analysis. A differential transformation is involved, and the transformation exploited the algebraic form of $F$ (and specifically the form of the relevant Jacobian). The $1$-norm of the transformed Jacobian is shown to be less than one, and thus exponential convergence to a unique fixed point is ensured, for all $x(0) \in \wt{\Delta}_n$.

\subsection{Proof of a Unique Fixed Point Which Is Locally Exponentially Stable}

Before we provide the result establishing there is a single fixed point, and further that it is locally exponentially stable,  we compute the Jacobian of $F:\mathbb R^n\rightarrow \mathbb R^n$ and then the related Jacobian of $F:\Delta_n\rightarrow\Delta_n$ in a coordinate basis we define, and establish some properties of the two Jacobians. For convenience, and when there is no risk of confusion, we drop the argument $k$ from $x(k)$ and $x$ from $\alpha(x(k))$. It is straightforward to obtain that
\begin{align}
\frac{\partial F_i}{\partial x_i} & = \frac{\gamma_i \alpha }{(1 - x_i)^2} - \frac{\gamma_i^2 \alpha^2 }{(1-x_i)^3} \nonumber \\
& = F_i \frac{1 - F_i}{1 - x_i} \label{eq:J_diag}
\end{align}
Similarly, we obtain, for $j \neq i$, 
\begin{align}
\frac{\partial F_i}{\partial x_j} & = - \frac{\gamma_i \gamma_j \alpha^2 }{(1-x_i)(1 - x_j)^2} \nonumber \\
& = - \frac{F_i F_j }{1 - x_j}\label{eq:J_offdiag}
\end{align}

We now show as a preliminary calculation that the corners of the simplex $\Delta_n$ are unstable equilibria for all social networks that are not star graphs, and the argument simultaneously allows us to show that $F$ is of class $\mathcal{C}^\infty$.
\begin{lemma}
Suppose that $0< \gamma_i < 1/2$, i.e. $\mathcal{G}$ is strongly connected but is not a star graph. Then $x = e_i$, where $e_i$ is the $i^{th}$ canonical unit vector, is an unstable equilibrium of \eqref{eq:nonlinearupdate} with map $F$ given in \eqref{eq:map_F_DF}.
\end{lemma}
\begin{proof}
Without loss of generality, consider $i = 1$. Observe that 
\begin{equation}
F_1(x)=\frac{\gamma_1}{\gamma_1+\sum_{i=2}^n\frac{(1-x_1)\gamma_i}{1-x_i}}
\end{equation}
and for $j\neq 1$, 
\begin{equation}
F_j(x)=\frac{\gamma_j(1-x_1)}{(1-x_j)(\gamma_1+\sum_{i=2}^n\frac{\gamma_i(1-x_1)}{1-x_k})}
\end{equation}
and it is evident that these expressions are analytic in $x_1$ for all $x\in\Delta_n$, and indeed for an open set enclosing $\Delta_n$. The same is then necessarily true of all their derivatives. Hence we conclude that $F$ is of class $\mathcal{C}^\infty$ in $\Delta_n$. 

At $x=e_1$, the expressions above yield that $F(e_1)=e_1$ and differentiating the expressions yields a value for the Jacobian at $x=e_1$ in which 
$\frac{\partial F_1}{\partial x_1} = \frac{1-\gamma_1}{\gamma_1}$, $\frac{\partial F_i}{\partial x_1} = -\frac{\gamma_i}{\gamma_1}$, $\frac{\partial F_i}{\partial x_j} = 0$ for all $i,j \neq 1$. 
It follows that the Jacobian $\frac{\partial F}{\partial x}$ associated with $F:\mathbb R^n\rightarrow\mathbb R^n$ at the point $x=e_1$ has a single eigenvalue at $(1-\gamma_1)/\gamma_1$ and all other eigenvalues are $0$.   Since $\gamma_1 < 1/2$, then $(1-\gamma_1)/\gamma_1 > 1$ and the fixed point $e_1$ is unstable. The associated eigenvector is $e_1$. This eigenvector has a nonzero projection onto $\Delta_n$ so that the instability is also an instability of the fixed point of $F:\Delta_n\rightarrow\Delta_n$. No matter what $(n-1)$-vector coordinatization we use for $\Delta_n$, the representation of $dF_x$ will be an $(n-1)\times(n-1)$ matrix with an eigenvalue greater than 1. 
\end{proof}
Since $e_i$ for all $i = 1, \hdots, n$ are unstable equilibria, we exclude them by defining an entity, distinct from $\Delta_n$, as $\wt{\Delta}_n = \{x_i : \sum_i^n x_i = 1, 0 \leq x_i \leq 1-\delta\}$, where $\delta>0$ is sufficiently small to ensure that any fixed point of $F$ in $\Delta_n$, save the unstable $e_i$, is contained in $\wt{\Delta}_n$. This ensures that $\wt{\Delta}_n$ is a compact, convex, and oriented manifold, which will allow us to use the results developed in Section~\ref{section:mainresult}. In other words, we now study the map in \eqref{eq:map_F_DF} as $F : \wt{\Delta}_n \mapsto \wt{\Delta}_n$. Now as already noted the above computed $n\times n$ Jacobian $\frac{\partial F}{\partial x}$, with elements given in \eqref{eq:J_diag} and \eqref{eq:J_offdiag}, is in fact not what we require to apply Theorem~\ref{thm:main_lefschetz_result}. This is because $\frac{\partial F}{\partial x}$ is the Jacobian computed in the coordinates of the Euclidean space  in which $\wt{\Delta}_n$ is embedded. We require the Jacobian \emph{on the manifold} $\wt{\Delta}_n$, which we will now obtain. We introduce a new coordinate basis $y\in \mathbb{R}^{n-1}$ where $y_1 = x_1, y_2 = x_2, \hdots, y_{n-1} = x_{n-1}$, and thus \emph{on the manifold $\wt{\Delta}_n$} we have $x_n = 1 - \sum_{k = 1}^{n-1} y_k$. On the manifold, and in the new coordinates, we define $G$ as the map with $G_1(y) = F_1(x), \hdots, G_{n-1}(y) = F_{n-1}(x)$, which means that $F_n = 1 - \sum_{k=1}^{n-1} G_k$. The Jacobian on the manifold of $\wt{\Delta}_n$ is in fact $dG_y$, which we now compute. For any $G_i(y_1, \hdots, y_{n-1}) = F_i(y_1, \hdots, y_{n-1}, 1 - \sum_{k=1}^{n-1} y_k)$, we have by the Chain rule that:
\begin{align}
\frac{\partial G_i}{\partial y_j} & = \sum_{k=1}^{n} \frac{\partial F_i}{\partial x_k}\frac{\partial x_k}{\partial y_j} \\
& = \frac{\partial F_i}{\partial x_j}\frac{\partial x_j}{\partial y_j} + \frac{\partial F_i}{\partial x_n}\frac{\partial x_n}{\partial y_j}
\end{align}
because $\partial x_k/\partial y_j = 0$ for $k \neq j, n$. In fact, we have from the definition of $y$, $\partial x_j/\partial y_j = 1$ and $\partial x_n/\partial y_j = - 1$. Thus, 
\begin{equation}
\frac{\partial G_i}{\partial y_j} = \frac{\partial F_i}{\partial x_j} - \frac{\partial F_i}{\partial x_n}
\end{equation}
In matrix form, it is straightforward to show that
\begin{align}\label{eq:dGy_form}
& \begin{bmatrix} \frac{\partial G_1}{\partial y_1} & \cdots & \frac{\partial G_1}{\partial y_{n-1}} \\ \vdots & \ddots & \vdots \\ \frac{\partial G_{n-1}}{\partial y_1} & \cdots & \frac{\partial G_{n-1}}{\partial y_{n-1}} \end{bmatrix}  = \nonumber \\
& \qquad \begin{bmatrix} \frac{\partial F_1}{\partial x_1} & \cdots & \frac{\partial F_1}{\partial x_n} \\ \vdots & \ddots & \vdots \\ \frac{\partial F_{n-1}}{\partial x_1} & \cdots & \frac{\partial F_{n-1}}{\partial x_n} \end{bmatrix} \begin{bmatrix} I_{n-1} \\ -\vect{1}_{n-1}^\top \end{bmatrix}
\end{align}
where $I_{n-1}$ is the $n-1$ dimensional identity matrix and $\vect{1}_{n-1}$ is the $n-1$ dimensional column vector of all ones.

Before we introduce the main result of this section, we state a linear algebra result which will be used in the proof.

\begin{lemma}[Corollary 7.6.2 in \cite{horn2012matrixbook}]\label{lem:AB_real}
Let $A, B \in \mathbb{R}^{n\times n}$ be symmetric. If $A$ is positive definite, then $AB$ is diagonalizable and has real eigenvalues. If, in addition, $B$ is positive definite or positive semidefinite, then the eigenvalues of $AB$ are all strictly positive or nonnegative, respectively.
\end{lemma}

\begin{thm}\label{thm:df_application}
Suppose that $\gamma_i < 1/2$ for all $i$. Then the map $F$ given in \eqref{eq:map_F_DF} has a unique fixed point in $\wt{\Delta}_n$, and this fixed point is locally exponentially stable.
\end{thm}
\begin{proof}
While we will need to use $dG_y$, for convenience we begin by studying certain properties of $\frac{\partial F}{\partial x}$, because it has certain properties which allow for easier delivery of specific conclusions in relation to $dG_y$. In summary, we will prove that  at any fixed point $\bar x \in \wt{\Delta}_n$, $\frac{\partial F}{\partial x}$ has a single eigenvalue at zero and all other eigenvalues are real, positive, and with magnitude less than one. We then show that the eigenvalues of $dG_y$ are the nonzero eigenvalues of $\frac{\partial F}{\partial x}$, which allows us to use Theorem~\ref{thm:main_lefschetz_result}.

Let us denote an arbitrary fixed point of $F$ as $\bar x$. Then clearly $F_i(\bar x) = \bar x_i$ for any $i$. Then it is straightforward to obtain that 
\begin{align}
\left.\frac{\partial F_i}{\partial x_i}\right|_{\bar x} & = \bar x_i \\
\left.\frac{\partial F_i}{\partial x_j}\right|_{\bar x} & = -\frac{\bar x_i \bar x_j}{1-\bar x_j}
\end{align}
In addition, it was shown in \cite{jia2015opinion_SIAM,ye2017DF_IFAC} that $\bar x_i > 0$ for all $i$. Since $\bar x \in \wt{\Delta}_n$, we immediately conclude that the diagonal entries of $\frac{\partial F}{\partial x}|_{\bar x}$ are strictly positive and the off-diagonal entries strictly negative. Moreover, it is straightforward to verify using \eqref{eq:J_diag} and \eqref{eq:J_offdiag} that the column sum of $\frac{\partial F}{\partial x}$ is equal to zero for every column. In other words, $\left[\frac{\partial F}{\partial x}\right]^\top$ is the Laplacian matrix associated with a strongly connected graph, which implies that $\frac{\partial F}{\partial x}$ has a single eigenvalue at zero and all other eigenvalues have positive real part \cite{ren2007distributed_book}. We now show that the other eigenvalues are strictly real and less than one in magnitude.

Define $A = \text{diag}[1- \bar x_i]$ as a diagonal matrix with the $i^{th}$ diagonal entry being $1- \bar x_i$. Since $\bar x_i \in \wt{\Delta}_n$, all diagonal entries of $A$ are strictly positive. The matrix $B = \frac{\partial F}{\partial x} A$ is symmetric, with diagonal entry $b_{ii} = \bar x_i(1-\bar x_i) > 0$ and off-diagonal entries $b_{ij} = -\bar x_i \bar x_j < 0$. Verify that, for any $i$, there holds 
\begin{align}
\sum_{j=1}^n b_{ij} & = \bar x_i(1 - \bar x_i) - \bar x_i\sum_{j=1, j\neq i}^n \bar x_j \\
& = \bar x_i (1 - \bar x_i - \sum_{j=1, j\neq i}^n \bar x_j) \\
& = 0
\end{align}
where the last equality was obtained by using the fact that $\bar x \in \wt{\Delta}_n \Leftrightarrow \sum_{j=1}^n \bar x_j = 1 \Leftrightarrow 1 - \bar x_i = \sum_{j=1, j\neq i}^n \bar x_j$. In other words, the row sum of $B$ is equal to zero for every row. It follows that $B$ is the Laplacian matrix of an undirected, complete graph; $B$ has a single eigenvalue at zero and all other eigenvalues are positive real \cite{ren2007distributed_book}. Using Lemma~\ref{lem:AB_real}, we thus conclude that $\frac{\partial F}{\partial x}|_{\bar x} = (A^{-1}B)^\top$ has all real eigenvalues (because $A^{-1}$ is positive definite and $B$ is positive semidefinite). Notice that $\text{trace}(\frac{\partial F}{\partial x}|_{\bar x}) = \sum_{i=1}^n \bar x_i = 1 = \sum_{j=1} \lambda_j(\frac{\partial F}{\partial x}|_{\bar x})$, where $\lambda_j$ is an eigenvalue of $\frac{\partial F}{\partial x}|_{\bar x}$. Since $n \geq 3$ and $\frac{\partial F}{\partial x}|_{\bar x}$ has only a single zero eigenvalue, it follows that all other eigenvalues of $\frac{\partial F}{\partial x}|_{\bar x}$ are strictly less than one (and real).

Define the matrix 
\begin{equation}
T = \begin{bmatrix} I_{n-1} & \vect{0}_{n-1} \\ -\vect{1}_{n-1}^\top & 1 \end{bmatrix}\,,\quad T^{-1} = \begin{bmatrix} I_{n-1} & \vect{0}_{n-1} \\ \vect{1}_{n-1}^\top & 1 \end{bmatrix}
\end{equation}
where $\vect{0}_{n-1}$ is the $n-1$ dimensional vector of all zeros. We established earlier that $\frac{\partial F}{\partial x}$ has column sum equal to zero, i.e. $\vect{1}^\top \frac{\partial F}{\partial x} = \vect{0}^\top$. Combining this column sum fact with \eqref{eq:dGy_form}, observe then, that
\begin{align}\label{eq:dGy_eigs}
\begin{bmatrix} dG_y & \frac{\partial F}{\partial x_n} \\ \vect{0}_{n-1}^\top & 0 \end{bmatrix} = T^{-1} \frac{\partial F}{\partial x} T
\end{align}
where $\frac{\partial F}{\partial x_n}$ is a column vector with $i^{th}$ element $\frac{\partial F_i}{\partial x_n}$. The similarity transform in \eqref{eq:dGy_eigs} tells us that the matrix on the left of \eqref{eq:dGy_eigs} has the same eigenvalues as $\frac{\partial F}{\partial x}$, and since the matrix is block triangular, it follows that $dG_y$ has the same nonzero eigenvalues as $\frac{\partial F}{\partial x}$.

Since we  assumed that $\bar x$ was an arbitrary fixed point  it follows that all eigenvalues of $\frac{\partial F}{\partial x}$ at any fixed point in $\wt{\Delta}_n$ are real and strictly less than one, which in turn implies that the eigenvalues of $dG_y$, at any fixed point $\bar y = [\bar x_1, \hdots, \bar x_{n-1}]^\top$, are inside the unit circle. By Theorem~\ref{thm:main_lefschetz_result}, $G$ has a unique fixed point $\bar y$ in $\wt{\Delta}_n$, and thus $F$ in \eqref{eq:map_F_DF} has a unique fixed point $\bar x$ in $\wt{\Delta}_n$. 
\end{proof}

\begin{remark}
The fact that the eigenvalues of $dG_y$ at a point in $\wt{\Delta}_n$ are a subset of those of $\frac{\partial F}{\partial x}$ is no surprise. Because $\wt{\Delta}_n$ is invariant under $F$, the translation of the affine space enclosing the set to define a linear space (including the origin) must have the property that this linear space is an invariant subspace for $\frac{\partial F}{\partial x}$. As such, linear algebra tells us that the eigenvalues of $\frac{\partial F}{\partial x}$ restricted to the invariant subspace are a subset of the full set of eigenvalues of $\frac{\partial F}{\partial x}$. We have chosen above to give a more ``explicit'' proof of the relation, in the process identifying the eigenvalue of $\frac{\partial F}{\partial x}$  that drops out in restricting to the invariant subspace.
\end{remark}

\begin{remark}
We note that it is straightforward to prove $\frac{\partial F}{\partial x}$, {\bf{for all}} $x\in \wt{\Delta}_n$, has strictly real eigenvalues with a single zero eigenvalue and all others positive. This property holds not only at the fixed point of $F$. However, via simulations, we have observed that the eigenvalues of $\frac{\partial F}{\partial x}$ can be greater than one (other than at a fixed point), and this may occur near the boundary of $\wt{\Delta}_n$. In \cite{ye2017DF_journal_arxiv}, the authors were therefore motivated to introduce a differential coordinate transform and show the transformed Jacobian had $1$-norm strictly less than one; nonlinear contraction analysis was then used to conclude exponential convergence to a unique fixed point $\bar x$. It is not always assured that such a transform exists; the one proposed in \cite{ye2017DF_journal_arxiv} and the proof of the norm upper bound were nontrivial and not intuitive. In this paper, we have greatly simplified the analysis by looking  at the Jacobian at only the fixed points of $F$ (which we initially assumed were not unique). However, the method of this paper guarantees only {\bf{local}} convergence, in the sense that although there can be only one fixed point, the existence of trajectories which are not convergent to a fixed point but rather for example converge to an orbit is not precluded.  Moreover, the technique of this paper cannot be easily extended to treat non-autonomous versions of \eqref{eq:nonlinearupdate}, which in this example application, occur when the social network topology is dynamic. For the nonautonomous case, the paper \cite{ye2017DF_journal_arxiv} uses the techniques of nonlinear contraction to conclude there is a \textbf{unique limiting trajectory $\bar x(k)$}, see \cite[Section IV]{ye2017DF_journal_arxiv} for details.
\end{remark}

\section{Concluding Remarks}\label{section:concludingremarks}
 This paper studies discrete-time nonlinear maps. We show that if the map operates in a compact, oriented manifold and the map itself is homotopically equivalent to the identity map (a condition satisfied if the manifold is convex) then evaluation of the Jacobian at the fixed points of the map can yield substantial results. Specifically, if the Jacobian has eigenvalues strictly inside the unit disk at all fixed points, then the map has a unique fixed point, and the fixed point is locally exponentially stable. This result is proved using Lefschetz fixed point theory. We then apply this result to a recent problem in social network analysis, simplifying existing proofs.  Future work will focus on unifying the Lefschetz approach by obtaining a similar result with a proof using Morse theory; preliminary results are encouraging. In addition, we will seek to determine whether any additional properties of $F$ would be needed to conclude a global convergence result.

\section*{Acknowledgements}
The work of M. Ye and B.D.O. Anderson was supported by the Australian Research Council (ARC) under grants \mbox{DP-130103610} and \mbox{DP-160104500}, and by Data61@CSIRO (formerly NICTA Ltd.). M. Ye was supported by an Australian Government Research Training Program (RTP) Scholarship. The authors would like to thank Jochen Trumpf for his discussion on Lefschetz and Morse theory.

\bibliographystyle{IEEEtran}
\bibliography{MYE_ANU,fixedpoint}
%

%
%
%




\end{document}